\documentclass[11pt]{article}
\usepackage{amsmath,amssymb,amsfonts,amsthm,latexsym,epsfig,hyperref,array,enumerate,geometry,graphicx,url,tikz}
\usepackage[all]{xy}
\newtheorem{theorem}{\bf Theorem}
\newtheorem{lemma}[theorem]{\bf Lemma}
\newtheorem{proposition}[theorem]{\bf Proposition}
\newtheorem{corollary}[theorem]{\bf Corollary}
\newtheorem{remark}[theorem]{\bf Remark}
\newtheorem{definition}[theorem]{\bf Definition}
\newtheorem{example}[theorem]{\bf Example}
\numberwithin{equation}{section}

\usepackage [english]{babel}
\usepackage [autostyle, english = american]{csquotes}
\MakeOuterQuote{"}

\begin{document}
\title{On some cryptographic properties of Boolean functions and their second-order derivatives}
\author{A. Musukwa, M. Sala and M. Zaninelli}
\date{}
\maketitle              
\begin{center}
	University of Trento, Via Sommarive, 14, 38123 Povo, Trento, Italy\ \\\{augustinemusukwa, maxsalacodes, zaninelli.marco21\}@gmail.com
\end{center}
\begin{abstract}
\noindent In this paper some cryptographic properties of Boolean functions, including weight, balancedness and nonlinearity, are studied, particularly focusing on splitting functions and cubic Boolean functions. Moreover, we present some quantities derived from the behaviour of second-order derivatives which allow us to determine whether a quadratic or cubic function is APN. \\ 

\end{abstract}
\section{Introduction}
Boolean functions are widely studied and applied in coding theory, cryptography and other fields. The properties of (vectorial) Boolean functions play a critical role in cryptography, particularly in the design of symmetric key algorithms in block cipher and nonlinear filters and combiners in stream ciphers. To find a function with good properties for a cryptosystem to possess robust resistance against most of the known attacks, a lot of effort is required. Cryptographic Boolean functions should satisfy various criteria simultaneously, for instance balancedness, high nonlinearity and good autocorrelation properties, to particularly resist linear cryptanalysis and differential cryptanalysis \cite{Tang}. 

This paper discusses some cryptographic properties of Boolean functions and is organised as follows. Section~\ref{sect-2} reports some known results which form a basis for what is being studied in this paper. In Section \ref{sect-3},  we show how the weight of any Boolean function can be related to the weights of some other functions at a lower dimension, we prove some results on weight and balancedness of "splitting" functions and a special class of cubic Boolean functions. In Section~\ref{sect-4}, we give an inequality relation which relate the nonlinearity of any Boolean function to the nonlinearity of some other functions at a lower dimension and nonlinearity of some splitting functions is proved. Finally, in Section~\ref{sect-5}, a parameter of a Boolean function based on its second-order derivatives is introduced and has been used for characterization of quadratic and cubic APN functions.

\section{Preliminaries}\label{sect-2}
In this section we report some definitions and results which we use in our work. For more details, the reader is referred to \cite{Beth,Bra,Car1,Chee,Mac,Wu}.

We denote the field of two elements, $0$ and $1$, by $\mathbb{F}$.  Any vector in the vector space $\mathbb{F}^n$ is denoted as $v$ (not as $\mathbf{v}$ or $\vec{v}$). The size of a set $A$ is denoted as $|A|$.

A {\em Boolean function (Bf)} is any function $f$ from $\mathbb{F}^n$ to $\mathbb{F}$ and a {\em vectorial Boolean function (vBf)} is any function $F$ from $\mathbb{F}^n$ to $\mathbb{F}^m$, $n,m\in\mathbb{N}$. However, in the present paper we only consider vBf's from $\mathbb{F}^n$ to $\mathbb{F}^n$. We use algebraic normal form (ANF for short), to represent the Bf's, which is the $n$-variable polynomial representation over $\mathbb{F}$ given by \[f(x_1,...,x_n)=\sum_{I\subseteq \mathcal{P}}a_I\left(\prod_{i\in I}x_i\right)\] where $\mathcal{P}=\{1,...,n\}$ and $a_I\in \mathbb{F}$. The {\em algebraic degree} or simply {\em degree} of $f$ (denoted by $\deg(f)$) is  \(\max_{a_I\ne 0}|I|.\) The set of all Bf's is denoted by $B_n$.

For a Bf $f$, we say that $f$ is {\em linear} if $\deg(f)\le 1$ and $f(0)=0$, {\em affine} if $\deg(f)\le 1$, {\em quadratic} if $\deg(f)=2$ and {\em cubic} if $\deg(f)=3$. The set of all affine functions is denoted by $A_n$. Given a vBf $F=(f_1,...,f_n)$, the functions  $f_1,...,f_n$ are called {\em coordinate functions} and the functions $\lambda\cdot F$, where $\lambda\in \mathbb{F}^n\setminus\{0\}$ and "$\cdot$" is a dot product, are called {\em component functions} and we write $F_\lambda=\lambda\cdot F$. The degree of a vBf $F$ is given by $\deg(F)=\max_{\lambda\in \mathbb{F}^n\setminus\{0\}}\deg(F_\lambda)$. We say that $F$ is {\em quadratic} if $\deg(F)=2$ and {\em cubic} if $\deg(F)=3$. If all components of a vBf $F$ are quadratic, we call $F$ a {\em pure quadratic}. 

For $m<n$, if $f$ is in $B_n$ and depends only on $m$ variables, then we denote by $f_{\restriction\mathbb{F}^m}$ its restriction to these $m$ variables. Clearly, $f_{\restriction\mathbb{F}^m}$ is in $B_m$. The {\em Hamming weight} of $f$ is given by $\mathrm{w}(f)=|\{x\in \mathbb{F}^n\mid f(x)=1\}|$. We say that $f$ is {\em balanced} if $\mathrm{w}(f)=2^{n-1}$. All non-constant affine functions are balanced. The {\em distance} between $f$ and $g$ is $d(f,g)=\mathrm{w}(f+g)$ and the {\em nonlinearity} of $f$  is $\mathcal{N}(f)=\min_{\alpha \in A_n}d(f,\alpha)$.

We define the {\em Walsh transform} of $f$, the function $\mathcal{W}_f$ from $\mathbb{F}^n$ to $\mathbb{Z}$, as \[\mathcal{W}_f(a)=\sum_{x\in\mathbb{F}^n}(-1)^{f(x)+a\cdot x}\,,\] for all $a \in \mathbb{F}^n$. Let $\mathcal{
	L}(f)=\max\limits_{a\in \mathbb{F}^n}|\mathcal{W}_f(a)|$.
We define $\mathcal{F}(f)$ as
\[\mathcal{F}(f)=\mathcal{W}_f(0)=\sum_{x\in \mathbb{F}^n}(-1)^{f(x)}=2^n-2\mathrm{w}(f).\]
Observe that $f$ is balanced if and only if $\mathcal{F}(f)=0$.
 
The nonlinearity of a Bf $f$ can also be given by \(\mathcal{N}(f)=2^{n-1}-\frac{1}{2}\mathcal{L}(f).\) A Bf $f$ on $n$ variables is called {\em bent} if \(\mathcal{N}(f)=2^{n-1}-2^{\frac{n}{2}-1}\) (this happens only for $n$ even). The lowest possible value for $\mathcal{L}(f)$ is $2^{\frac{n}{2}}$ and this bound is achieved for bent functions (and only them).

Let $a,b\in \mathbb{F}^n$. The {\em first-order derivative} of $f\in B_n$ at $a$ is defined by \[D_af(x)=f(x+a)+f(x)\] and its {\em second-order derivative} at $a$ and $b$ is \[D_bD_af(x)=f(x)+f(x+b)+f(x+a)+f(x+a+b).\]

\begin{theorem}\label{bent-thm}
	A Bf $f$ on $n$ variables is bent if and only if $D_af$ is balanced for any nonzero $a\in\mathbb{F}^n$.
\end{theorem}

For $n$ odd, a Bf $f$ is called {\em semi-bent} if $\mathcal{N}(f)=2^{n-1}-2^{\frac{n-1}{2}}$. A vBf $F$ in odd dimension is {\em almost-bent (AB)} if all its components are  semi-bent.

\begin{theorem}\label{permuatation}
	Let $F$ be a vBf. Then $F$ is a {\em permutation} if and only if all components are balanced.
\end{theorem}

Two Bf's $f,g:\mathbb{F}^n\rightarrow \mathbb{F}$ are said to be {\em affine equivalent} if there exists an affinity $\varphi:\mathbb{F}^n\rightarrow \mathbb{F}^n$ such that $f=g\circ \varphi$. This relation is denoted by $\sim_A $ and we write $f\sim_A g$. Observe that $\sim_A$ is an equivalence relation. The following result is well-known.

\begin{proposition}\label{equivalence-properties}
	Let $f,g\in B_n$ be such that $f\sim_A g$ . Then $\mathrm{w}(f)=\mathrm{w}(g)$ and so $f$ is balanced $\iff g$ is balanced.
\end{proposition}

\begin{remark}\label{fourier-equivalence}
	Since, by Proposition \ref{equivalence-properties}, $\mathrm{w}(f)=\mathrm{w}(g)$ if $f\sim_A g$, then it also implies that $\mathcal{F}(f)=\mathcal{F}(g)$ as $\mathcal{F}(f)=2^n-2\mathrm{w}(f)$.
\end{remark}

\begin{proposition}\label{extended-Walsh}
	Let $f,g\in B_n$ be such that $f\sim_A g$. Then \[\{|\mathcal{W}_f(a)|\}_{a\in \mathbb{F}^n}=\{|\mathcal{W}_g(a)|\}_{a\in \mathbb{F}^n}.\] Moreover, we have $\mathcal{N}(f)=\mathcal{N}(g)$. 
\end{proposition}

Next we present the theorem on classification of quadratic Boolean functions, whose proof can be found in \cite{Mac} page 438.

\begin{theorem}\label{quadratic}
	Let $f\in B_n$ be quadratic. Then
	
	\begin{itemize}
		\item[(i)] $f\sim_A x_1x_2+\cdots x_{2k-1}x_{2k}+x_{2i+1}$ with $k\leq \lfloor \frac{n-1}{2}\rfloor$ if $f$ is balanced,
		\item[(ii)] $f\sim_A x_1x_2+\cdots x_{2k-1}x_{2k}+c$, with $k\leq \lfloor \frac{n}{2}\rfloor$ and $c\in \mathbb{F}$, if $f$ is unbalanced.
	\end{itemize}
\end{theorem}

The proof of the next theorem and lemma (respectively) can be found in \cite{Cus} on page 134. 

\begin{theorem}\label{nonlinearity-quadratics}
	Let $f$ be a quadratic Bf denoted as in Theorem \ref{quadratic}. Then we have \\$\mathcal{W}_f(a)\in\{0,\pm 2^{n-k}\}$, for $a\in \mathbb{F}^n$, and \(\mathcal{N}(f)=2^{n-1}-2^{n-k-1}\).
\end{theorem} 

\begin{lemma}\label{weight-affine-equivalent-quadratics}
	Two quadratic Bf's $g$ and $h$ on $\mathbb{F}^n$ are affine equivalent if and only if $\mathrm{w}(g)=\mathrm{w}(h)$ and $\mathcal{N}(g)=\mathcal{N}(h)$.
\end{lemma}

An element $a\in \mathbb{F}^n$ is called a {\em linear structure} of Bf $f$ if $D_af$ is constant. Denote by $V(f)=\{a\in\mathbb{F}^n\mid D_af \text{ is a constant}\}$ the set of all linear structures and we call it the {\em linear space} of a Bf $f$. A Bf $f$ is {\em partially-bent} if there exists a linear subspace $W$ of $\mathbb{F}^n$ such that the restriction of $f$ to $W$ is affine and the restriction of $f$ to any complementary subspace $U$ of $W$, $W\oplus U=\mathbb{F}^n$, is bent \cite{Cal}. In fact the linear subspace $W$ of $\mathbb{F}^n$ is formed by the set of all linear structures of $f$, that is, $W=V(f)$ and observe that the dimension of $U$ must be even. A partially-bent $f$ can be represented as a direct sum of the restricted functions such as $f(y+z)=f(y)+f(z)$, for all $y\in U$ and $z \in V(f)$. 

\begin{remark}\label{quadratic-partially-bent}
 From Theorem \ref{quadratic}, it can be deduced that any quadratic function $f$ is partially-bent and $\dim V(f)$ is even if $n$ is even and odd if $n$ is odd. Moreover, we must have $\dim V(f)=0$ if and only if $f$ is bent.
\end{remark} 

The following result is well-known.

\begin{proposition}\label{balanced-splitting}
	A Bf $g(x_1,...,x_{n-1})+x_n$ on $n$ variables is balanced.
\end{proposition}

\section{On the weight of Boolean functions}\label{sect-3}
In this section we classify the weight of some class of cubic functions and others. Some conditions for these functions to be balanced are determined. 

\begin{definition}\label{splitting}
	A Bf $f$ on $n$ variables is a {\em splitting function} if \[f\sim_A g(x_1,...,x_s)+h(x_{s+1},...,x_n),\] for some positive integer $s<n$, $g\in B_s$ and $h\in B_{n-s}$.
\end{definition}

\begin{remark}\label{split-weight-fourier}
	If $g(x_1,...,x_s)$, with $s<n$, is in $B_n$ then $\mathrm{w}(g)=2^{n-s}\mathrm{w}(g_{\restriction\mathbb{F}^s})$ and \(\mathcal{F}(g)=2^{n-s}\mathcal{F}(g_{\restriction\mathbb{F}^s})\). Furthermore, $g$ is balanced if and only if  $g_{\restriction\mathbb{F}^s}$ is balanced and also \(\mathcal{F}(g)=0\) if and only if \(\mathcal{F}(g_{\restriction\mathbb{F}^s})=0\).
\end{remark}

Next we consider the weight and balancedness of splitting Bf's.
\begin{lemma}\label{fourier-split}
	Let $f\in B_n$ be such that \(f\sim_Ag(x_1,...,x_s)+h(x_{s+1},...,x_n)\), with $s<n$ Then \[\mathcal{F}(f)=\mathcal{F}(g_{\restriction\mathbb{F}^s})\mathcal{F}(h_{\restriction\mathbb{F}^{n-s}})=2^{-n}\mathcal{F}(g)\mathcal{F}(h).\]
\end{lemma}

\begin{proof}
	Since, by Remark \ref{fourier-equivalence}, $\mathcal{F}(f)$ is invariant under affine equivalence, then we have
	\begin{align*}
	\mathcal{F}(f)&=\sum_{(y,x)\in\mathbb{F}^s\times\mathbb{F}^{n-s}}(-1)^{g(y)+h(x)}=\sum_{y\in \mathbb{F}^s}(-1)^{g(y)}\sum_{x\in \mathbb{F}^{n-s}}(-1)^{h(x)}\\&=\mathcal{F}(g_{\restriction\mathbb{F}^s})\mathcal{F}(h_{\restriction\mathbb{F}^{n-s}})=2^{-n}\left(2^{n-s}\mathcal{F}(g_{\restriction\mathbb{F}^s})\right)\left(2^s\mathcal{F}(h_{\restriction\mathbb{F}^{n-s}})\right)\\&=2^{-n}\mathcal{F}(g)\mathcal{F}(h). \qedhere
	\end{align*}
\end{proof}

It is immediate from Remark \ref{split-weight-fourier} and Lemma \ref{fourier-split} that the following corollary holds. 

\begin{corollary}\label{general-fourier-split}
	For $t\in\mathbb{N}$ and $1\leq i\leq t$, let $X_i\subset X=\{x_1,...,x_n\}$, with $|X_i|=n_i$, be such that all $X_i$ are pairwise disjoint. If \(f(X)=\sum_{i=1}^tf_i(X_i),\) with $f_i~\in~B_{n_i}$, then \(\mathcal{F}(f)=2^{n-r}\prod_{i=1}^t\mathcal{F}({f_i}_{\restriction\mathbb{F}^{n_i}})\), with $r=n_1+\cdots +n_t$. 
\end{corollary}

\begin{proposition}\label{split-weight}
	Let  $f\in B_n$ be such that $f\sim_Ag(x_1,...,x_s)+h(x_{s+1},...,x_n)$, with $s<n$. Then \begin{align*}\mathrm{w}(f)&=2^{n-s}\mathrm{w}(g_{\restriction\mathbb{F}^s})+2^s\mathrm{w}(h_{\restriction\mathbb{F}^{n-s}})-2\mathrm{w}(g_{\restriction\mathbb{F}^s})\mathrm{w}(h_{\restriction\mathbb{F}^{n-s}})\\&=\mathrm{w}(g)+\mathrm{w}(h)-2^{1-n}\mathrm{w}(g)\mathrm{w}(h).\end{align*}
\end{proposition}

\begin{proof}
	We have 
\begin{align*}
	\mathrm{w}(f)&=2^{n-1}-\frac{1}{2}\mathcal{F}(f)=2^{n-1}-\frac{1}{2}\left(2^{-n}\mathcal{F}(g_{\restriction\mathbb{F}^s})\mathcal{F}(h_{\restriction\mathbb{F}^{n-s}})\right)\nonumber\\&=2^{n-1}-\frac{1}{2}\left[\left(2^s-2\mathrm{w}(g_{\restriction\mathbb{F}^s})\right)\left(2^{n-s}-2\mathrm{w}(h_{\restriction\mathbb{F}^{n-s}})\right)\right]\nonumber\\&=2^{n-s}\mathrm{w}(g_{\restriction\mathbb{F}^s})+2^s\mathrm{w}(h_{\restriction\mathbb{F}^{n-s}})-2\mathrm{w}(g_{\restriction\mathbb{F}^s})\mathrm{w}(h_{\restriction\mathbb{F}^{n-s}})\nonumber\\&=2^{n-s}\mathrm{w}(g_{\restriction\mathbb{F}^s})+2^s\mathrm{w}(h_{\restriction\mathbb{F}^{n-s}})-2^{1-n}\left(2^{n-s}\mathrm{w}(g_{\restriction\mathbb{F}^s})\right)\left(2^s\mathrm{w}(h_{\restriction\mathbb{F}^{n-s}})\right)\nonumber\\&=\mathrm{w}(g)+\mathrm{w}(h)-2^{1-n}\mathrm{w}(g)\mathrm{w}(h).\qedhere
\end{align*}
\end{proof}

We now present some results on balanced splitting functions.

\begin{theorem}\label{split-balanced}
	Let $f\in B_n$ be such that $f\sim_A g(x_1,...,x_s)+h(x_{s+1},...,x_n)$, with $s<n$. Then $f$ is balanced if and only if either $g$ or $h$ is balanced. 
\end{theorem}

\begin{proof}
	$f$ is balanced $\Longleftrightarrow \mathcal{F}(f)=0\Longleftrightarrow \left(\mathcal{F}(g_{\restriction\mathbb{F}^s})=0 \text{ or } \mathcal{F}(h_{\restriction\mathbb{F}^{n-s}})=0\right)\Longleftrightarrow$ either $g$ or $h$ is balanced. 
\end{proof}

\begin{proposition}\label{split-terms-of-same-degree-weight}
	Let $f\in B_n$, with $\deg(f)=m$, be such that \[f\sim_A \sum_{i=0}^{k-1}\prod_{j=1}^{m}x_{mi+j}.\] Then $\mathcal{F}(f)=2^{n-mk}(2^m-2)^k$ and \(\mathrm{w}(f)=2^{n-1}-2^{n-mk-1}(2^m-2)^k\).
\end{proposition}

\begin{proof}
	First, let $f_i=\prod_{j=1}^{m}x_{mi+j}$ so that $f\sim_A \sum_{i=0}^{k-1}f_i$. Then, by Corollary~\ref{general-fourier-split}, we have  $\mathcal{F}(f)=2^{n-mk}\prod_{i=0}^{k-1}\mathcal{F}({f_i}_{\restriction{\mathbb{F}^m}})$. For all $x\in\mathbb{F}^m\setminus\{\mathbf{1}\}$,  observe that $f_i{_{\restriction{\mathbb{F}^m}}}(x)=0$, and $f_i{_{\restriction{\mathbb{F}^m}}}(\mathbf{1})=1$ so $\mathcal{F}({f_i}_{\restriction{\mathbb{F}^m}})=2^m-2$.  Thus $\mathcal{F}(f)=2^{n-mk}(2^m-2)^k$. So $\mathrm{w}(f)=2^{n-1}-\frac{1}{2}\mathcal{F}(f)=2^{n-1}-\frac{1}{2}[2^{n-mk}(2^m-2)^k]=2^{n-1}-2^{n-mk-1}(2^m-2)^k$.
\end{proof}

Observe that the function $f$ in Proposition \ref{split-terms-of-same-degree-weight} is balanced if and only if $m=1$, that is, $f$ is balanced if and only if it is a linear function.

\begin{remark}\label{quadratic_weight}
	All quadratic Bf's are splitting functions (deduced from Theorem~\ref{quadratic}) and those which are unbalanced are of the form given in Proposition~\ref{split-terms-of-same-degree-weight} and their complements, with $m=2$. So applying Proposition~\ref{split-terms-of-same-degree-weight}, $\mathrm{w}(f)=2^{n-1}-2^{n-k-1}$ and \\$\mathrm{w}(f+1)=2^{n-1}+2^{n-k-1}$. This result on weight of quadratic is well-known.
\end{remark}

Now we study the weight and balancedness of Bf's in some given form. We show how the weight of a Bf on $n$ variables can be related to the weights of some other functions at a lower dimension.

Any Bf can be expressed in the form  \begin{align}\label{factoring} f\sim_A x_1g(x_2,...,x_n)+h(x_2,...,x_n).\end{align} Observe that \(f\sim_A x_1g(x_2,...,x_n)+h(x_2,...,x_n)=x_1(g+h)+(1+x_1)h\). So any Bf $f$ on $n+1$ variables can be written in the form \begin{align}\label{convolutional-product-form}f\sim_A x_{n+1}g(x_1,...,x_n)+(1+x_{n+1})h(x_1,...,x_n).\end{align} We say that $f$ is the {\em convolutional product} of $g$ and $h$. Observe that the convolutional product is a special case of the form defined by \begin{align}\label{generalised-convolutional-product} f\sim_A \left(\prod_{j=1}^mx_j\right)g(x_{m+1},...,x_{m+n})+\left(1+\prod_{j=1}^mx_j\right)h(x_{m+1},...,x_{m+n}),\end{align} for some positive integer $m$ and Bf's $g$ and $h$ on $n$ variables. In fact, for any Bf $f$, there exists a positive integer $m$ such that $f$ can be expressed in the form \eqref{generalised-convolutional-product}. Next we show that if the weights of $g$ and $h$ are known, then the weight of $f$ is obtained.

\begin{theorem}\label{generalised-convolutional_product_weight}
	Let $f\in B_{m+n}$ be a Bf of the form (\ref{generalised-convolutional-product}). Then \begin{itemize}
		\item[(a)] \(\mathrm{w}(f)=(2^m-1)\mathrm{w}(h_{\restriction\mathbb{F}^n})+\mathrm{w}(g_{\restriction\mathbb{F}^n})\),
		\item[(b)] $f$ is balanced if and only if \(\mathcal{F}(h_{\restriction\mathbb{F}^n})=-\mathcal{F}(g_{\restriction\mathbb{F}^n})/(2^m-1)\),
		\item[(b)] $f$ is balanced if both $g$ and $h$ are balanced,
		\item[(d)] $f$ is unbalanced if one in $\{g,h\}$ is balanced and the other is not.
	\end{itemize}  	
\end{theorem}

\begin{proof} We have \\\(f=\left(\prod_{j=1}^mx_j\right)g(x_{m+1},...,x_{m+n})+\left(1+\prod_{j=1}^mx_j\right)h(x_{m+1},...,x_{m+n}).\)
\begin{itemize}
		\item[(a)] Let $X=(x,y)\in\mathbb{F}^m\times\mathbb{F}^n$. Then we have 
	\begin{align}\label{fourier}
	\mathcal{F}(f)&=\sum_{X\in\mathbb{F}^{n+m}}(-1)^{f(X)}=\sum_{(x,y)\in\mathbb{F}^m\setminus \{1\}\times\mathbb{F}^n}(-1)^{h(y)}+\sum_{(x,y)\in\{1\}\times\mathbb{F}^n}(-1)^{g(y)}\nonumber\\&=(2^m-1)\sum_{y\in\mathbb{F}^n}(-1)^{h(y)}+\sum_{y\in\mathbb{F}^n}(-1)^{g(y)}\nonumber\\&=(2^m-1)\mathcal{F}(h_{\restriction\mathbb{F}^n})+\mathcal{F}(g_{\restriction\mathbb{F}^n})
	\end{align}
	 Since $\mathcal{F}(f)=2^{m+n}-2w(f)$, so we have 
	\begin{align*}
	w(f)&=2^{n+m-1}-\frac{1}{2}\mathcal{F}(f)= 2^{n+m-1}-\frac{1}{2}\left[(2^m-1)\mathcal{F}(h_{\restriction\mathbb{F}^n})+\mathcal{F}(g_{\restriction\mathbb{F}^n})\right]\\&=2^{n+m-1}-\frac{1}{2}\left[(2^m-1)(2^n-2w(h_{\restriction\mathbb{F}^n}))+(2^n-2w(g_{\restriction\mathbb{F}^n}))\right]\\&= 2^{n+m-1}-\frac{1}{2}\left[2^{n+m}-2^{m+1}w(h_{\restriction\mathbb{F}^n})+2w(h_{\restriction\mathbb{F}^n})-2w(g_{\restriction\mathbb{F}^n})\right]\\&=(2^m-1)w(h_{\restriction\mathbb{F}^n})+w(g_{\restriction\mathbb{F}^n}).
	\end{align*}
	
	\item[(b)] Recall that $f$ is balanced  if and only if we have $\mathcal{F}(f)=0$ if and only if \\$(2^m-1)\mathcal{F}(h_{\restriction\mathbb{F}^n})+\mathcal{F}(g_{\restriction\mathbb{F}^n})=0$ if and only if $\mathcal{F}(h_{\restriction\mathbb{F}^n})=-\mathcal{F}(g_{\restriction\mathbb{F}^n})/(2^m-1)$.
	
	\item[(c)] Suppose $g$ and $h$ are both balanced. Then $\mathcal{F}(g_{\restriction\mathbb{F}^n})=\mathcal{F}(h_{\restriction\mathbb{F}^n})=0$. Applying Equation \eqref{fourier}, it implies that $\mathcal{F}(f)=0$, and so $f$ is balanced.
	
	\item[(d)] Without loss of generality, suppose that $g$ is balanced while $h$ not. Then $\mathcal{F}(g_{\restriction\mathbb{F}^n})=0$ and $\mathcal{F}(h_{\restriction\mathbb{F}^n})\neq 0$ which, by Equation \eqref{fourier}, implies that $\mathcal{F}(f)\neq 0$, and so $f$ is unbalanced.\qedhere
\end{itemize}   
\end{proof}

\begin{remark}\label{convolutional_product_weight}
	If $m=1$ in Theorem \ref{generalised-convolutional_product_weight} [i.e., $f=(x_{n+1})g+(1+x_{n+1})h$] then we have $\mathrm{w}(f)=\mathrm{w}(h_{\restriction\mathbb{F}^n})+\mathrm{w}(g_{\restriction\mathbb{F}^n})$.
\end{remark}
Observe that if $g$ and $h$ are such that $g=h\circ \varphi +1$, for some affinity $\varphi$, then $f$ is balanced since $\mathrm{w}(f)=\mathrm{w}(h_{\restriction\mathbb{F}^n})+\mathrm{w}(g_{\restriction\mathbb{F}^n})=\mathrm{w}(h_{\restriction\mathbb{F}^n})+2^n-\mathrm{w}(h\circ \varphi_{\restriction\mathbb{F}^n})=\mathrm{w}(h_{\restriction\mathbb{F}^n})+2^n-\mathrm{w}(h_{\restriction\mathbb{F}^n})=2^n$. 

Finally, we consider the weight of cubic Bf's. Generally, it is difficult to determine the weight for Bf's of degree greater than $2$ (see \cite{Car4}). Next we present a result which completely describes the weight of a special class of cubic functions. This result allows us to construct an algorithm that computes the weight of any cubic function. 

Since we have the knowledge of weights of affine and quadratic functions (see Remark \ref{quadratic_weight}) and by applying Remark \ref{convolutional_product_weight}, we can state our classification theorem for the weight of the special class of cubic functions. We omit the proof of Theorem~\ref{weight-of-cubic} because it is a direct case-by-case computation.

\begin{theorem}\label{weight-of-cubic}
	Let $f=x_{n+1}g(x_1,...,x_n)+(1+x_{n+1})h(x_1,...,x_n)$ be a cubic Bf such that $\deg(g),\deg(h)\leq 2$. Then $h\sim_A q=x_1x_2+\cdots+x_{2k-1}x_{2k}$ or $h\sim_A \bar{q}=q+1$ and $g\sim_A r=x_1x_2+\cdots+x_{2\ell-1}x_{2\ell}$ or $g\sim_A\bar{r}=r+1$, with $k, \ell\leq \lfloor \frac{n}{2}\rfloor$, if $h$ and $g$ are quadratic unbalanced.  Moreover, 
	{\small \begingroup
		\allowdisplaybreaks\[\mathrm{w}(f)=\begin{cases}2^n & \text{if both } h \text{ and } g \text{ are balanced}\\ 2^{n-1}  & \text{if } h \text{ (resp. $g$) is bal. quad. and } g \text{ (resp. $h$)}=0\\ 2^n+2^{n-1}  & \text{if } h \text{ (resp. $g$) is bal. quad. and } g \text{ (resp. $h$)}= 1\\2^{n-1}\pm 2^{n-k-1}  & \text{if } h \text{ is unbal. quad. and } g=0\\2^n+2^{n-1}\pm 2^{n-k-1}  & \text{if } h \text{ is unbal. quad. and } g= 1\\ 2^{n-1}\pm 2^{n-\ell-1}  & \text{ if } h=0 \text{ and } g \text{ is unbal. quad. } \\2^n+2^{n-1}\pm 2^{n-\ell-1}   & \text{if } h=1 \text{ and } g \text{ is unbal. quad. } \\2^n\pm 2^{n-k-1}  & \text{if } h \text{ is unbal. quad. and } g \text{ is bal. }\\ 2^n\pm 2^{n-\ell-1}  & \text{if } h \text{ is bal. and } g \text{ is unbal. quad. }\\2^n-2^{n-k-1}-2^{n-\ell-1}   & \text{if } h\sim_A q \text{ and } g\sim_A r \\2^n+2^{n-k-1}+2^{n-\ell-1}   & \text{if } h\sim_A \overline{q} \text{ and } g\sim_A \overline{r}\\2^n+2^{n-k-1}-2^{n-\ell-1}  & \text{if } h\sim_A \overline{q} \text{ and } g \sim_A r\\2^n-2^{n-k-1}+2^{n-\ell-1}  & \text{if } h\sim_A q \text{ and } g \sim_A \overline{r}.\end{cases}\]\endgroup}
\end{theorem}

Thanks to Theorem \ref{weight-of-cubic}, the following corollary which gives a description of all balanced cubic functions of the class $f=x_{n+1}g(x_1,...,x_n)+(1+x_{n+1})h(x_1,...,x_n)$, with $\deg(g),\deg(h)\leq 2$, is deduced. 

\begin{corollary}\label{balanced-convolutional-cubic}
	With the same notation as in Theorem \ref{weight-of-cubic}, a cubic Bf $f$ is balanced if and only if one of the following holds: \begin{itemize}
		\item[(a)] both $g$ and $h$ are balanced, 
		\item[(b)] $g\sim_A q$ and $h\sim_A \overline{q}$, 
		\item[(c)] $g\sim_A \overline{q}$ and $h\sim_A q$.
	\end{itemize}
\end{corollary}

Applying Lemma \ref{weight-affine-equivalent-quadratics} and Theorem \ref{weight-of-cubic}, Corollary \ref{balanced-convolutional-cubic} can be simplified as in the following.

\begin{corollary}
	Let $f=x_{n+1}g(x_1,...,x_n)+(1+x_{n+1})h(x_1,...,x_n)$, with $g,h\in B_n$ and $\deg(h),\deg(g)\leq 2$, be cubic Boolean function. Then $f$ is balanced if and only if either both $g$ and $h$ are balanced or $g=h\circ \varphi +1$, for some affinity $\varphi$.
\end{corollary}

\begin{proof}
By Equation \eqref{fourier}, we have \(\mathcal{F}(f)=\mathcal{F}(g_{\restriction\mathbb{F}^n})+\mathcal{F}(h_{\restriction\mathbb{F}^n})\).  So it follows that $f$ is balanced $\iff \mathcal{F}(f)=0\iff \mathcal{F}(g_{\restriction\mathbb{F}^n})=-\mathcal{F}(h_{\restriction\mathbb{F}^n})\iff 2^n-2\mathrm{w}(g_{\restriction\mathbb{F}^n})=-2^n+2\mathrm{w}(h_{\restriction\mathbb{F}^n})\iff \mathrm{w}(g_{\restriction\mathbb{F}^n})+\mathrm{w}(h_{\restriction\mathbb{F}^n})=2^n\iff$  either both $g$ and $h$ are balanced or $\mathrm{w}(g_{\restriction\mathbb{F}^n})=\mathrm{w}(h_{\restriction\mathbb{F}^n}+1)$ with both  $g$ and $h$ unbalanced quadratics $\iff$ either both $g$ and $h$ are balanced or $g=h\circ \varphi +1$, for some affinity $\varphi$ (by Lemma \ref{weight-affine-equivalent-quadratics}).
\end{proof}

Now we consider cubic Bf's which cannot be expressed in the form described in Theorem \ref{weight-of-cubic}. If a Bf $f$ is expressed in the form \eqref{factoring}, that is, $f=x_1g(x_2,...,x_n)+h(x_2,...,x_n)$ then $\mathrm{w}(f)=\mathrm{w}((g+h)_{\restriction\mathbb{F}^{n-1}})+\mathrm{w}(h_{\restriction\mathbb{F}^{n-1}})$. Since our interest is in cubic functions, it can be assumed that $g$ is quadratic and $h$ can be affine, quadratic or cubic. If $h$ is affine or quadratic, then weight of $f$ can be easily computed by Theorem \ref{weight-of-cubic}. It becomes difficult  to find the weight of $f$ if $h$ is cubic since in this case it implies that $g+h$ is also cubic and finding $\mathrm{w}(h_{\restriction\mathbb{F}^{n-1}})$ and $\mathrm{w}((g+h)_{\restriction\mathbb{F}^{n-1}})$ is not easy. However, we can recursively repeat the process of decomposition of $f$ so that its weight is the sum of weights of some affine or quadratic functions on a vector space of dimension $<n$ over $\mathbb{F}$. For instance, further decomposing $g+h$ and $h$ into the form $g+h=x_2g_1(x_3,...,x_n)+h_1(x_3,...,x_n)$ and $h=x_2g'_1(x_3,...,x_n)+h'_1(x_3,...,x_n)$, the weight of $f$ becomes $\mathrm{w}(f)=\mathrm{w}((g_1+h_1)_{\restriction\mathbb{F}^{n-2}})+\mathrm{w}({h_1}_{\restriction\mathbb{F}^{n-2}})+\mathrm{w}((g'_1+h'_1)_{\restriction\mathbb{F}^{n-2}})+\mathrm{w}({h'_1}_{\restriction\mathbb{F}^{n-2}})$. We use this idea to build an algorithm which computes the weight of any cubic Bf's and its efficiency and simplicity relies on Theorem \ref{weight-of-cubic} and the known results about the weights of affine and quadratic functions.

\subsection*{Algorithm 1}\label{algorithm1}
The following algorithm computes the weight of a cubic function $f$ on $n$ variables:

\begin{tabular}{ll}
{\bf Input: }	&  cubic function $f$, \\ 
{\bf Output: }&  $\mathrm{w}(f)$,\\ 
{\bf Step 1: }&  express $f$ in the form $f=x_1g(x_2,...,x_n)+h(x_2,...,x_n)$ so that\\
& $g$ is quadratic,\\ 
{\bf Step 2: }&  if $\deg(h)\leq 2$, compute $\mathrm{w}(f)$ by using Theorem \ref{weight-of-cubic} and return $\mathrm{w}(f)$,\\ 
{\bf Step 3: }&  otherwise, recursively compute the weights of $g+h$ and $h$ by \\
&applying {\bf Step 1} and {\bf Step 2},\\ 
{\bf Step 4: }&  sum up all the weights found to obtain $\mathrm{w}(f)$.
\end{tabular}

\section{Nonlinearity of Boolean functions}\label{sect-4}
We begin with the nonlinearity of a function whose terms have the degree but their variables are pairwise disjoint.

\begin{proposition}\label{terms-of-same-degree-nonlinearity}
	Let $f\in B_n$, with $\deg(f)=m$ and $m>1$, be such that \[f\sim_A\sum_{t=0}^{k-1}\prod_{j=1}^{m}x_{mt+j}.\] Then \(\mathcal{N}(f)=2^{n-1}-2^{n-mk-1}(2^m-2)^k.\)
\end{proposition}

\begin{proof}
	Let $f_i=\prod_{j=1}^{m}x_{mi+j}$. Then $f\sim_A \sum_{i=0}^{k-1}f_i$. Let $l_\alpha(x)=\alpha\cdot x$, where $\alpha,x\in\mathbb{F}^n$. Observe that $f+l_\alpha$ is balanced if $l_\alpha$ has some variables which are not in $f$ (see Proposition \ref{balanced-splitting}) and in this case, we have $\mathcal{W}_{f}(\alpha)=\mathcal{F}(f+l_\alpha)=0$. Thus we can assume that $l_\alpha(x)=l_a(X)=a\cdot X$, with \(a=(a_0,...,a_{k-1})\) and $X=(y_0,...,y_{k-1})$ in  $\left(\mathbb{F}^m\right)^k$, so that all variables in $l_a$ are also in $f$. By Corollary~\ref{general-fourier-split}, we have \[\mathcal{W}_{f}(\alpha)=\mathcal{F}(f+l_a)=2^{n-mk}\prod_{i=0}^{k-1}\mathcal{F}([f_i+l_{a_i}]_{\restriction\mathbb{F}^m}).\] Recall that $\mathcal{N}(f)=2^{n-1}-\frac{1}{2}\max_{\alpha\in\mathbb{F}^n}|\mathcal{W}_{f}(\alpha)|$. Clearly, $|\mathcal{W}_{f}(\alpha)|$ is maximal if all \\$\mathcal{F}([g_i+l_{a_i}]_{\restriction\mathbb{F}^m})$ are maximal. $\mathcal{F}([f_i+l_{a_i}]_{\restriction\mathbb{F}^m})=2^m-2\mathrm{w}([f_i+l_{a_i}]_{\restriction\mathbb{F}^m})$ and it is clear that $\mathrm{w}([f_i+l_{a_i}]_{\restriction\mathbb{F}^m})\neq 0$. So $\mathcal{F}([f_i+l_{a_i}]_{\restriction\mathbb{F}^m})$ is maximal if $a_i=(0,...,0)$ since in this case $\mathrm{w}([f_i+l_{a_i}]_{\restriction\mathbb{F}^m})=\mathrm{w}({f_i}_{\restriction\mathbb{F}^m})=1$. Thus, $|\mathcal{W}_{f}(\alpha)|$ is maximal if, for all $i$, we have $\mathcal{F}([f_i+l_{a_i}]_{\restriction\mathbb{F}^m})=\mathcal{F}({f_i}_{\restriction\mathbb{F}^m})=2^m-2$, implying that it is maximal when $\alpha=(0,...,0)$. Substituting $\mathcal{F}([f_i+l_{a_i}]_{\restriction\mathbb{F}^m})=2^m-2$, we obtain $\mathcal{W}_{f}(\alpha)=2^{n-mk}(2^m-2)^k$. Hence \(\mathcal{N}(f)=2^{n-1}-2^{n-mk-1}(2^m-2)^k.\)
\end{proof}

\begin{remark}
	We deduce from Proposition \ref{terms-of-same-degree-nonlinearity} that $f$ is bent if and only if $m=2$ and $k=n/2$, for $n$ even, otherwise $2^{n-mk-1}2^k(2^{m-1}-1)^k$ would be equal to $2^{\frac{n}{2}-1}$, for some positive integer $k$, contradicting the fact that $(2^{m-1}-1)\nmid 2^{\frac{n}{2}-1}$ since $(2^{m-1}-1)$ is odd and  $2^{\frac{n}{2}-1}$ cannot be divisible by an odd number. 
\end{remark} 

\begin{theorem}\label{walsh-nonlinearity-generalised-conv}
	Let $f$ be a Bf of the form (\ref{generalised-convolutional-product}). Let $\alpha=(a,b)\in \mathbb{F}^m\times \mathbb{F}^n$, with $a=(a_1,...,a_m)$ and $b=(b_1,...,b_n)$ . Then 
	\begin{itemize}
		\item[(i)]\( \mathcal{W}_f(\alpha)=
		\begin{cases}
		\left(2^m-1\right)\mathcal{W}_{h_{\restriction\mathbb{F}^n}}(b)+\mathcal{W}_{g_{\restriction\mathbb{F}^n}}(b) & \text{ if } a=0\\(-1)^\lambda\left(\mathcal{W}_{g_{\restriction\mathbb{F}^n}}(b)-\mathcal{W}_{h_{\restriction\mathbb{F}^n}}(b)\right)& \text{ otherwise},
		\end{cases}\)  \\with $\lambda=a_1+\cdots +a_m$,
		\item [(ii)] \(\mathcal{N}(f)\geq (2^m-1)\mathcal{N}(h_{\restriction\mathbb{F}^n})+\mathcal{N}(g_{\restriction\mathbb{F}^n})\).	
	\end{itemize}
\end{theorem}

\begin{proof} Since nonlinearity is invariant under affine equivalence, we can simply write\\ \(f=\left(\prod_{j=1}^mx_j\right)g(x_{m+1},...,x_{m+n})+\left(1+\prod_{j=1}^mx_j\right)h(x_{m+1},...,x_{m+n}).\)
	Let\\ $X=(y,x)\in \mathbb{F}^m\times \mathbb{F}^n$, with $y=(x_1,...,x_m)$ and $x=(x_{m+1},...,x_{m+n})$, and $\mathbf{1}=(1,1,...,1)$. Then 
	\begingroup
	\allowdisplaybreaks
	\begin{align*}
	\mathcal{W}_f(\alpha)&=\sum_{X\in\mathbb{F}^{m+n}}(-1)^{f(X)+\alpha\cdot X}\\&=\sum_{(y,x)\in\mathbb{F}^m\setminus\{\mathbf{1}\}\times \mathbb{F}^n}(-1)^{h(x)+a\cdot y+b\cdot x}+\sum_{(y,x)\in\{\mathbf{1}\}\times \mathbb{F}^n}(-1)^{g(x)+a\cdot y+b\cdot x}\\&=\sum_{(y,x)\in\mathbb{F}^m\times \mathbb{F}^n}(-1)^{h(x)+a\cdot y+b\cdot x}-\sum_{(y,x)\in\{\mathbf{1}\}\times \mathbb{F}^n}(-1)^{h(x)+a\cdot y+b\cdot x}\\&+\sum_{(y,x)\in\{\mathbf{1}\}\times \mathbb{F}^n}(-1)^{g(x)+a\cdot y+b\cdot x}\\&=\left( \sum_{y\in\mathbb{F}^m}(-1)^{a\cdot y}\right)\cdot\left(\sum_{x\in\mathbb{F}^n}(-1)^{h(x)+b\cdot x}\right)-(-1)^\lambda\sum_{x\in\mathbb{F}^n}(-1)^{h(x)+b\cdot x}\\&+(-1)^\lambda\sum_{x\in \mathbb{F}^n}(-1)^{g(x)+b\cdot x}\\&=\left( \sum_{y\in\mathbb{F}^m}(-1)^{a\cdot y}\right)\mathcal{W}_{h_{\restriction\mathbb{F}^n}}(b)-(-1)^\lambda\mathcal{W}_{h_{\restriction\mathbb{F}^n}}(b)+(-1)^\lambda\mathcal{W}_{g_{\restriction\mathbb{F}^n}}(b)\\&=\begin{cases} (2^m-1)\mathcal{W}_{h_{\restriction\mathbb{F}^n}}(b)+\mathcal{W}_{g_{\restriction\mathbb{F}^n}}(b)& \text{ if } a=0\\ (-1)^\lambda\left[\mathcal{W}_{g_{\restriction\mathbb{F}^n}}(b)-\mathcal{W}_{h_{\restriction\mathbb{F}^n}}(b)\right]& \text{ otherwise}.
	\end{cases}
	\end{align*}
	\endgroup
	To reach the last step we used the fact that \[\sum_{y\in\mathbb{F}^m}(-1)^{a\cdot y}=\begin{cases}
	2^m & \text{ if } a=0,\\0 & \text{ otherwise}
	\end{cases}\]
	and also that $\lambda=0$ if $a=0$.
	
		For any two integers $c$ and $d$, it is well-known that $|c+d|\leq |c|+|d|$. Clearly, we have  
	
	\[|\mathcal{W}_f(\alpha)|\leq \begin{cases} (2^m-1)|\mathcal{W}_{h_{\restriction\mathbb{F}^n}}(b)|+|\mathcal{W}_{g_{\restriction\mathbb{F}^n}}(b)|& \text{ if } a=0\\ |\mathcal{W}_{g_{\restriction\mathbb{F}^n}}(b)|+|\mathcal{W}_{h_{\restriction\mathbb{F}^n}}(b)|& \text{ otherwise}. \end{cases}\] Since \[|\mathcal{W}_{g_{\restriction\mathbb{F}^n}}(b)|+|\mathcal{W}_{h_{\restriction\mathbb{F}^n}}(b)|\leq (2^m-1)|\mathcal{W}_{h_{\restriction\mathbb{F}^n}}(b)|+|\mathcal{W}_{g_{\restriction\mathbb{F}^n}}(b)|,\] then we deduce that, for any $\alpha=(a,b)$, we have 
	\[|\mathcal{W}(\alpha)|\leq (2^m-1)|\mathcal{W}_{h_{\restriction\mathbb{F}^n}}(b)|+|\mathcal{W}_{g_{\restriction\mathbb{F}^n}}(b)|.\]
	So
\begin{align*}
	\mathcal{N}(f)&=2^{n+m-1}-\frac{1}{2}\max_{\alpha\in\mathbb{F}^{n+m}}|\mathcal{W}_f(\alpha)|\\&\geq 2^{n+m-1}-\frac{1}{2}\max_{b\in \mathbb{F}^n}\left((2^m-1)|\mathcal{W}_{h_{\restriction\mathbb{F}^n}}(b)|+|\mathcal{W}_{g_{\restriction\mathbb{F}^n}}(b)|\right)\\&\geq 2^{n+m-1}-\frac{1}{2}(2^m-1)\max_{b\in \mathbb{F}^n}|\mathcal{W}_{h_{\restriction\mathbb{F}^n}}(b)|-\frac{1}{2}\max_{b\in \mathbb{F}^n}|\mathcal{W}_{g_{\restriction\mathbb{F}^n}}(b)|\\&=(2^m-1)2^{n-1}+2^{n-1}-\frac{1}{2}(2^m-1)\max_{b\in \mathbb{F}^n}|\mathcal{W}_{h_{\restriction\mathbb{F}^n}}(b)|-\frac{1}{2}\max_{b\in \mathbb{F}^n}|\mathcal{W}_{g_{\restriction\mathbb{F}^n}}(b)|\\&=(2^m-1)2^{n-1}-\frac{1}{2}(2^m-1)\max_{b\in \mathbb{F}^n}|\mathcal{W}_{h_{\restriction\mathbb{F}^n}}(b)|+2^{n-1}-\frac{1}{2}\max_{b\in \mathbb{F}^n}|\mathcal{W}_{g_{\restriction\mathbb{F}^n}}(b)|\\&=(2^m-1)\mathcal{N}(h_{\restriction\mathbb{F}^n})+\mathcal{N}(g_{\restriction\mathbb{F}^n}).\qedhere
	\end{align*}
\end{proof}

\begin{remark}\label{nonlinearity-remark}
	Note that if $m=1$, by Theorem \ref{walsh-nonlinearity-generalised-conv}, the nonlinearity of \[f\sim_A x_{n+1}g(x_1,..,x_n)+(1+x_{n+1})h(x_1,...,x_n)\] is $\mathcal{N}(f)\geq \mathcal{N}(h_{\restriction\mathbb{F}^n})+\mathcal{N}(g_{\restriction\mathbb{F}^n})$.
\end{remark}
It is immediate from Theorem \ref{nonlinearity-quadratics} and Remark \ref{nonlinearity-remark} that the following corollary holds.

\begin{corollary}\label{nonlinearity-cubics}
	Let $f$ be as described in Theorem \ref{weight-of-cubic}. Then \[\mathcal{N}(f)\geq\begin{cases} 2^{n-1}-2^{n-k-1} & \text{ if } g \text{ is quadratic and } h \text{ affine,}\\2^{n-1}-2^{n-\ell-1} & \text{ if } g \text{ is affine and } h \text{ quadratic,}\\2^n-2^{n-k-1}-2^{n-\ell-1} & \text{ if both } g \text{ and } h \text{ are quadratic}.
	\end{cases}\]
\end{corollary}
Corollary \ref{nonlinearity-cubics} suggests a way of constructing Bf's with high non-linearity.

\section{A Characterization of APN Functions}\label{sect-5}
In this section we define a parameter which is used for characterization of quadratic and cubic APN function. This parameter can also be used to describe some properties for quadratic and cubic partially-bent functions.

\subsection{Some known results on APN functions}
Some definitions and known results on APN functions, which can be found in \cite{Ber,Beth,Cal,Cant,Car1}, are reported.

\begin{definition}
	 Define \(\delta_F(a,b)=|\{x\in \mathbb{F}^n| D_aF(x)=b\}|\), for $a,b\in \mathbb{F}^n$ and vBf $F$. The {\em differential uniformity of $F$ is} \[\delta(F)=\max_{a\neq 0,b\in \mathbb{F}^n}\delta_F(a,b)\] and always satisfies $\delta(F)\geq 2$. A function with $\delta(F)=2$ is called {\em Almost Perfect Nonlinear (APN)}.
\end{definition}

For a vBf $F$, the $k$th power moment of Walsh transform is defined as \[L_k(F)=\sum_{\lambda\in\mathbb{F}^n\setminus\{0\}}L(F_\lambda).\] 

Next we state a result in which APN functions are characterized by the fourth power moment of Walsh transform. 
\begin{theorem}\label{APN-momentum-vectorial}
	Let $F$ be a vBf from $\mathbb{F}^n$ to itself. Then \[L_4(F)\geq 2^{3n+1}(2^n-1).\] Moreover, $F$ is APN if and only if equality holds.
\end{theorem}

The following result can be easily deduced from Theorem \ref{APN-momentum-vectorial}.

\begin{theorem}\label{APN-first-order-derivatives}
	Let $F:\mathbb{F}^n\rightarrow \mathbb{F}^n$ be a vBf. Then \[\sum_{\lambda\neq 0,a\in\mathbb{F}^n}\mathcal{F}^2(D_aF_\lambda)\geq 2^{2n+1}(2^n-1).\] Moreover, $F$ is APN if and only if equality holds. 
\end{theorem}

\subsection{The parameter $\mathcal{M}(f)$}\label{subsect-2}
We define and study some properties of a parameter for a Boolean function based on its second-order derivatives and in the next subsection we use it for characterization of quadratic and cubic APN functions.

\begin{definition}\label{mf}
	For $a\in\mathbb{F}^n$ and $f\in B_n$, define $\mathit{Z}_a(f):=\{b\in \mathbb{F}^n\mid D_bD_af= 0\}$, 
	$\mathit{U}_a(f):=\{b\in \mathbb{F}^n\mid D_bD_af= 1\}$ and $\mathcal{M}_a(f):=|\mathit{Z}_a(f)|-|\mathit{U}_a(f)|$. We define the parameter $\mathcal{M}(f)$ by  \[\mathcal{M}(f):=\sum\limits_{a\in \mathbb{F}^n\setminus\{0\}}\mathcal{M}_a(f).\]
\end{definition}

\begin{lemma}\label{equivalence-derivatives}
	Let $g_1,g_2\in B_n$ be such that $g_1=g_2(Mx+w)$, with invertible $M\in GL_n(\mathbb{F})$ and $w\in\mathbb{F}^n$. Then, for any $a\in\mathbb{F}^n\setminus\{0\}$, we have $D_ag_1\sim_A D_{M\cdot a}g_2$.
\end{lemma}

\begin{proof}
	Let $\varphi$ be the affinity of $\mathbb{F}^n$ associated with $M\in GL_n(\mathbb{F})$, $w\in\mathbb{F}^n$, that is, $\varphi(y)=M\cdot y+w$, for all $y\in \mathbb{F}^n$. For $a\in \mathbb{F}^n$, we have 
	\begin{align*}
	D_ag_1(x)&=D_a(g_2\circ \varphi)(x)\\&=g_2(\varphi(x+a))+g_2(\varphi(x))\\&=g_2(M\cdot (x+a)+w)+g_2(\varphi(x))\\&=g_2(M\cdot x+M\cdot a+w)+g_2(\varphi(x))\\&=g_2(M\cdot a+\varphi(x))+g_2(\varphi(x))\\&=D_{M\cdot a}g_2(\varphi(x))=(D_{M\cdot a}g_2\circ \varphi)(x).
	\end{align*}
	So it implies that $D_ag_1\sim_A D_{M\cdot a}g_2$.
\end{proof}

\begin{proposition}
	Let $f\in B_n$. Then, for all $a\in \mathbb{F}^n$,
	\begin{itemize}
		\item[(i)] $\mathit{Z}_a(f)$ is a vector space and has nonzero dimension,
		\item[(ii)] $\mathit{U}_a(f)$ is either a coset of $\mathit{Z}_a(f)$ or the empty set.
	\end{itemize}
\end{proposition}

\begin{proof}
	(i) It is clear that $0$ is in $Z_a(f)$ since $D_0D_a(f)=0$. Suppose we have $b_1,b_2\in Z_a(f)$. Then \[D_{b_1+b_2}D_af(x)=D_{b_1}f(x)+D_{b_2}D_af(x+b_1)=0+0=0,\] implying that $b_1+b_2\in Z_a(f)$ [note that we deduced that $D_{b_2}D_af(x+b_1)=0$ from Lemma \ref{equivalence-derivatives}]. To show that it is of nonzero dimension, observe that if $a=0$ then $Z_a(f)=\mathbb{F}^n$ and if $a\neq 0$, then we have $D_aD_af(x)=0$, implying that $\{0,a\}\subseteq Z_a(f)$. So the dimension of $Z_a(f)$ is at least $1$.
	
	(ii) Suppose that $U_a(f)\neq \varnothing$. For any $b_1\in U_a(f)$, we show that $b_1+Z_a(f)=U_a(f)$. Let $b_2=b_1+d$, with $d\in Z_a(f)$. We have \[D_{b_2}D_af(x)=D_{b_1+d}D_af(x)=D_{b_1}D_af(x)+D_dD_af(x+b_1)=1+0=1.\] Thus, $b_2\in U_a(f)$. Conversely, for $e\in U_a(f)$, we have \[D_{b_1+e}D_af(x)=D_{b_1}D_af(x)+D_eD_af(x+b_1)=1+1=0.\] It follows that $e+b_1\in Z_a(f)\implies e\in b_1+Z_a(f)$.
\end{proof}

\begin{proposition}\label{quadratic-cubic-mf}
	Let $f\in B_n$ be a Bf with $\deg(f)\in\{2,3\}$. Then, for some even integer $j$, with $1<j<n$ and any $a\in\mathbb{F}^n$, we have 
	\[\mathcal{M}_a(f)=
	\begin{cases}0 & \text{ if and only if $D_af$ balanced},\\
	2^n& \text{ if and only if $D_af$ is constant},\\
	2^{n-j} & \text{otherwise}.\\
	\end{cases}\] 
\end{proposition}

\begin{proof}
	Since $\deg(f)\in\{2,3\}$ then $\deg(D_af)\in\{0,1,2\}$. It is clear from the definition of $\mathcal{M}_a(f)$ that \(\deg(D_af)=0\iff \mathcal{M}_a(f)=2^n.\)
	
	Suppose that $\deg(D_af)=1$. Then $D_af(x)$ is a non-constant affine function, so it is balanced. That is, we can write $D_af(x)=v\cdot x+c$, for some $v\in\mathbb{F}^n\setminus\{0\}$ and $c\in\mathbb{F}$. Observe that 
	\begin{align*}
	D_bD_af(x)&=v\cdot x+c+v\cdot (x+b)+c\\&=v\cdot x+v\cdot x+v\cdot b\\&=v\cdot b.
	\end{align*} So $D_bD_af(x)=0\iff b\in W=<v>^{\perp}$ and $D_bD_af(x)=1\iff b\in W^c$ ($A^\perp$ denotes the dual set and $A^c$ denotes the complement of a set $A$). Thus, $Z_a(f)=W$ and $U_a(f)=W^c$. It is clear that $|W|=|W^c|=2^{n-1}$. So we have $\mathcal{M}_a(f)=0$.
	
	Finally, suppose that $\deg(D_af)=2$, that is, by Theorem \ref{quadratic}, we know that $D_af\sim_A x_1x_2+\cdots +x_{2i-1}x_{2i}+x_{2i+1}$, with $i \leq \lfloor(n-1)/2\rfloor$, if $D_af$ is balanced and $D_af\sim_A x_1x_2+\cdots +x_{2i-1}x_{2i}+e$, with $i \leq \lfloor n/2\rfloor$ and $e\in \mathbb{F}$, if $D_af$ is unbalanced. Suppose that $D_af$ is balanced. Then \[|Z_a(f)|=|\{c=(c_1,...,c_n)\in\mathbb{F}^n\mid c_1=\cdots=c_{2i+1}=0\}|\] and  \[|U_a(f)|=|\{c=(c_1,...,c_n)\in\mathbb{F}^n\mid c_1=\cdots=c_{2i}=0, c_{2i+1}=1\}|.\] Observe that in both cases, $|Z_a(f)|=|U_a(f)|=2^{n-2i-1}$. Hence $\mathcal{M}_a(f)=0$. Now suppose that $D_af$ is unbalanced. Then we have \[|Z_a(f)|=|\{c=(c_1,...,c_n)\in\mathbb{F}^n\mid c_1=\cdots=c_{2i}=0\}|\] and  \(U_a(f)=\varnothing\). It follows that $|Z_a(f)|=2^{n-2i}$ and $|U_a(f)|=0$.  So it implies that $\mathcal{M}_a(f)=2^{n-2i}$.
\end{proof}

\begin{proposition}\label{mf-partially-bent}
	For any quadratic and cubic partially-bent function $f$, we have \[\mathcal{M}(f)=2^n(2^k-1),\] where $k=\dim V(f)$.
\end{proposition}

\begin{proof}
	We know, from Proposition \ref{quadratic-cubic-mf}, that $\mathcal{M}_a(f)=0$ if and only if $D_af$ is balanced and $\mathcal{M}_a(f)=2^n$ if and only if $D_af$ is a constant. We deduce, from the definition, that for any partially-bent function $f$, $D_af$ is constant if and only if $a\in V(f)$ and $D_af$ is balanced if and only if $a\notin V(f)$. Recall that all quadratic functions are partially-bent. Thus, for any quadratic function or cubic partially-bent function $f$, we have \[\mathcal{M}(f)=\sum_{a\in \mathbb{F}^n\setminus\{0\}}\mathcal{M}_a(f)=\sum_{a\in V(f)\setminus\{0\}}\mathcal{M}_a(f)=2^n(2^k-1),\] with $k=\dim V(f)$. 
\end{proof}

If a function $f$ is bent, then $k=0$ and so, by Proposition \ref{mf-partially-bent}, $\mathcal{M}(f)=0$. Thus, we state this in the following.  

\begin{corollary}\label{bent-mf}
	Let $f\in B_n$ be a quadratic or cubic function. Then $f$ is bent if and only if $\mathcal{M}(f)=0$.
\end{corollary}
Observe that the result in Corollary \ref{bent-mf} can also be deduced by Theorem~\ref{bent-thm} and Proposition~\ref{quadratic-cubic-mf}.

\begin{lemma}\label{semi-bent-V(f)}
	Let $f\in B_n$, with $n$ odd, be quadratic. Then $\dim V(f)\geq 1$ and equality holds if and only if $f$ is semi-bent.
\end{lemma} 

\begin{proof}
	From Theorem \ref{quadratic}, observe that \[|V(f)|=|\{c=(c_1,...,c_n)\in\mathbb{F}^n| c_1=\cdots=c_{2i}=0, i\leq (n-1)/2\}|.\] It follows that $|V(f)|=2^{n-2i}$. Since $n$ is odd, so we must have $\dim V(f)\geq 1$.  It can be observed, from Theorem \ref{nonlinearity-quadratics}, that $f$ is semi-bent $\iff f\sim_A x_1x_2+\cdots +x_{n-2}x_{n-1}+x_{n}$  or $f\sim_A x_1x_2+\cdots +x_{n-2}x_{n-1}+c$, with $c\in\mathbb{F}$, from which we deduce that $f$ is semi-bent  $\iff \dim V(f)=1$.
\end{proof}

By Theorem \ref{quadratic} and Lemma \ref{semi-bent-V(f)}, the following corollary holds. 
\begin{corollary}\label{semi-bent-mf}
	For $n$ odd, a quadratic Bf $f$ is semi-bent if and only if $\mathcal{M}(f)=2^n$.
\end{corollary}

\subsection{APN functions and their second-order derivatives}
For a vBf $F:\mathbb{F}^n\rightarrow\mathbb{F}^n$, define \(\mathcal{M}(F)=\sum_{\lambda\neq 0\in \mathbb{F}^n}\mathcal{M}(F_\lambda).\)  It is clear from Subsection~\ref{subsect-2} that the quantity $\mathcal{M}(F)$ is defined based on second-order derivatives of components of $F$. We establish a connection between the fourth power moment of the Walsh transform and the value \(\mathcal{M}(F)\), and consequently derive a characterization of quadratic and cubic APN functions based on the latter quantity. 

\begin{lemma}\label{momentum-lemma-mf-vectorial}
	Let \(F:\mathbb{F}^n\rightarrow \mathbb{F}^n\) be a vBf of $\deg(F)\in\{2,3\}$. Then \[L_4(F)=2^{3n}(2^n-1)+2^{2n}\mathcal{M}(F).\]
\end{lemma}

\begin{proof}
We have
\begingroup
\allowdisplaybreaks
	\begin{align}\label{momentum-fourier-equation}
	L_4(F)&=\sum_{\lambda\neq 0\in\mathbb{F}^n}L_4(F_\lambda)=\sum_{\lambda\neq 0\in\mathbb{F}^n}\sum_{a\in \mathbb{F}^n}\mathcal{W}_{F_\lambda}^4 (a)\nonumber\\&=\sum_{\lambda\neq 0\in\mathbb{F}^n}\sum_{a\in\mathbb{F}^n}\sum_{x,y,z,w\in \mathbb{F}^n}(-1)^{F_\lambda(x)+F_\lambda(y)+F_\lambda(z)+F_\lambda(w)+a\cdot (x+y+z+w)}\nonumber\\&=\sum_{\lambda\neq 0\in\mathbb{F}^n}\sum_{a\in\mathbb{F}^n}\sum_{x,y,z,w\in \mathbb{F}^n}(-1)^{F_\lambda(x)+F_\lambda(y)+F_\lambda(z)+F_\lambda(w)}(-1)^{a\cdot (x+y+z+w)}\nonumber\\&=\sum_{\lambda\neq 0\in \mathbb{F}^n}\sum_{x,y,z,w\in\mathbb{F}^n}(-1)^{F_\lambda(x)+F_\lambda(y)+F_\lambda(z)+F_\lambda(w)}\sum_{a\in \mathbb{F}^n}(-1)^{a\cdot (x+y+z+w)}\nonumber\\&=\sum_{\lambda\neq 0\in \mathbb{F}^n}\sum_{x,y,z,w\in\mathbb{F}^n |x+y+z+w=0}2^n(-1)^{F_\lambda(x)+F_\lambda(y)+F_\lambda(z)+F_\lambda(w)}\nonumber\\&=2^n\sum_{\lambda\neq 0\in \mathbb{F}^n}\sum_{x,y,z,w\in\mathbb{F}^n |w=x+y+z}(-1)^{F_\lambda(x)+F_\lambda(y)+F_\lambda(z)+F_\lambda(w)}\nonumber\\&=2^n\sum_{\lambda\neq 0\in \mathbb{F}^n}\sum_{x,y,z\in\mathbb{F}^n }(-1)^{F_\lambda(x)+F_\lambda(y)+F_\lambda(z)+F_\lambda(x+y+z)}\nonumber\\&
	\left(\text{substituting $y=x+b$ and $z=x+c$ we have}\right)\nonumber\\& =2^n\sum_{\lambda\neq 0\in\mathbb{F}^n}\sum_{x,b,c\in \mathbb{F}^n}(-1)^{F_\lambda(x)+F_\lambda(x+b)+F_\lambda(x+c)+F_\lambda(x+b+c)}\nonumber\\&=2^n\sum_{\lambda\neq 0\in\mathbb{F}^n}\sum_{x,b,c\in \mathbb{F}^n}(-1)^{D_bF_\lambda(x)+D_bF_\lambda(x+c)}\nonumber\\&=2^n\sum_{\lambda\neq 0\in\mathbb{F}^n}\sum_{x,b,c\in \mathbb{F}^n}(-1)^{D_cD_bF_\lambda(x)}\\& (\deg(D_cD_bF_\lambda)=1\implies \sum_{x\in \mathbb{F}^n}(-1)^{D_cD_bF_\lambda(x)}=0, \text{ so we have}) \nonumber\\&=2^n\sum_{\lambda\neq 0\in\mathbb{F}^n}\sum_{x,b,c\in \mathbb{F}^n|\deg(D_cD_bF_\lambda)=0}(-1)^{D_cD_bF_\lambda(x)}\nonumber\\&=2^n\sum_{\lambda\neq 0\in\mathbb{F}^n}2^n\sum_{b,c\in \mathbb{F}^n|\deg(D_cD_bF_\lambda)=0}(-1)^{D_cD_bF_\lambda(0)}\nonumber\\&= 2^{2n}\sum_{\lambda\neq 0\in\mathbb{F}^n}\left(\sum_{b,c\in \mathbb{F}^n|D_cD_bF_\lambda=0}(-1)^0+\sum_{b,c\in \mathbb{F}^n|D_cD_bF_\lambda=1}(-1)^1\right)\nonumber\\&=2^{2n}\sum_{\lambda\neq 0\in\mathbb{F}^n}\left(|\{b,c\in \mathbb{F}^n\mid D_cD_bF_\lambda=0\}|-|\{b,c\in \mathbb{F}^n\mid D_cD_bF_\lambda=1\}|\right)\nonumber\\&=2^{2n}\sum_{\lambda\neq 0\in\mathbb{F}^n}\sum_{b\in\mathbb{F}^n}\left(|\{c\in \mathbb{F}^n\mid D_cD_bF_\lambda=0\}|-|\{c\in \mathbb{F}^n\mid D_cD_bF_\lambda=1\}|\right)\nonumber\\&=2^{2n}\sum_{\lambda\neq 0\in\mathbb{F}^n}\left(|\{c\in \mathbb{F}^n\mid D_cD_0F_\lambda=0\}|-|\{c\in \mathbb{F}^n\mid D_cD_0F_\lambda=1\}|\right)\nonumber\\&+2^{2n}\sum_{\lambda\neq 0\in\mathbb{F}^n}\sum_{b\neq 0\in\mathbb{F}^n}\mathcal{M}_b(F_\lambda)\nonumber\\&=2^{2n}\sum_{\lambda\neq 0\in\mathbb{F}^n}\left(|\{c\in \mathbb{F}^n\mid D_c(0)=0\}|-|\{c\in \mathbb{F}^n\mid D_c(0)=1\}|\right)\nonumber\\&+2^{2n}\sum_{\lambda\neq 0\in\mathbb{F}^n}\sum_{b\neq 0\in\mathbb{F}^n}\mathcal{M}_b(F_\lambda)\nonumber\\&=2^{2n}\sum_{\lambda\neq 0\in\mathbb{F}^n}\left(2^n-0\right)+2^{2n}\sum_{\lambda\neq 0\in\mathbb{F}^n}\sum_{b\neq 0\in\mathbb{F}^n}\mathcal{M}_b(F_\lambda)\nonumber\\&=2^{2n}\sum_{\lambda\neq 0\in\mathbb{F}^n}2^n+2^{2n}\sum_{\lambda\neq 0\in\mathbb{F}^n}\sum_{b\neq 0\in\mathbb{F}^n}\mathcal{M}_b(F_\lambda)\nonumber\\&=2^{3n}(2^n-1)+2^{2n}\sum_{\lambda\neq 0\in\mathbb{F}^n}\sum_{b\neq 0\in \mathbb{F}^n}\mathcal{M}_b(F_\lambda)\nonumber\\&=2^{3n}(2^n-1)+2^{2n}\sum_{\lambda\neq 0\in\mathbb{F}^n}\mathcal{M}(F_\lambda)=2^{3n}(2^n-1)+2^{2n}\mathcal{M}(F)\nonumber. \qedhere
\end{align}
\endgroup
\end{proof}

We deduce by Lemma \ref{momentum-lemma-mf-vectorial} and Theorem \ref{APN-momentum-vectorial} that the following theorem holds.

\begin{theorem}\label{mf-cubic-quadratic-APN}
	Let $F:\mathbb{F}^n\rightarrow \mathbb{F}^n$ be a vBf with $\deg(F)\in \{2,3\}$. Then \[\mathcal{M}(F)\geq 2^{n}(2^n-1).\] Moreover, $F$ is APN if and only if equality holds.
\end{theorem}

 By Theorem \ref{mf-cubic-quadratic-APN}, the following corollary holds.   

\begin{corollary}\label{APN-mf}
	If a vBf $F:\mathbb{F}^n\rightarrow \mathbb{F}^n$ is a quadratic or cubic APN then there is a nonzero  $\lambda\in \mathbb{F}^n$ such that $\mathcal{M}(F_\lambda)\leq 2^n$.
\end{corollary}

By Proposition \ref{mf-partially-bent}, we can deduce that the following corollary holds.

\begin{corollary}\label{sum-mf-quadratic}
	Let $F:\mathbb{F}^n\rightarrow \mathbb{F}^n$ be a quadratic function or cubic partially-bent function. Then \begin{align}\label{mf-quadratic-eqn}\mathcal{M}(F)=2^n\sum_{\lambda\in\mathbb{F}^n\setminus\{0\}} (2^{\dim V(F_\lambda)}-1).\end{align} 
\end{corollary}

\begin{example}
	{\rm Let $F(x_1, x_2, x_3) = (f_1, f_2, f_3)$ where $f_1 = x_1x_3 + x_2x_3+x_1$, $f_2= x_2x_3+x_1 + x_2$ and $f_3= x_1x_2+x_1 + x_2+x_3$ are all in $B_3$. One can verify that all components are quadratic. By Corollary \ref{sum-mf-quadratic}, $\mathcal{M}(F)=2^3\cdot (2^3-1)=56$ and so, by Theorem \ref{mf-cubic-quadratic-APN}, we conclude that $F$ is an APN function. Moreover, all components are balanced, implying that $F$ is an APN permutation.}
\end{example} 

We deduce, from Lemma \ref{semi-bent-V(f)}, Corollary \ref{sum-mf-quadratic} and Theorem \ref{mf-cubic-quadratic-APN}, that the following corollary holds.

\begin{corollary}\label{quadratic-mf-APN-odd dimension}
	Let $F:\mathbb{F}^n\rightarrow \mathbb{F}^n$, with $n$ odd, be a quadratic function or cubic partially-bent function. Then $F$ is APN if and only if, for all $\lambda\neq 0 \in\mathbb{F}^n$, $\mathcal{M}(F_\lambda)=2^n$. 
\end{corollary}

By Theorem \ref{mf-cubic-quadratic-APN} and Corollary \ref{sum-mf-quadratic},  the following result holds. 

\begin{corollary}\label{sum-mf-quadratic-APN}
	Let $F:\mathbb{F}^n\rightarrow \mathbb{F}^n$ be a quadratic function or cubic partially-bent function. Then \begin{align}\label{mf-quadratic-eqn-1}\sum_{\lambda\neq 0\in\mathbb{F}^n} (2^{\dim V(F_\lambda)}-1)\geq 2^n-1.\end{align} Moreover, equality holds if and only if $F$ is APN. 
\end{corollary}

By applying Lemma \ref{semi-bent-V(f)} and Corollary \ref{sum-mf-quadratic-APN}, we can deduce the only well-known result present in this subsection.

\begin{theorem}[\cite{Bud}]\label{quadratic-apn}
	Let $F:\mathbb{F}^n\rightarrow\mathbb{F}^n$, with $n$ odd, be a pure quadratic function. Then $F$ is APN if and only if it is AB. 
\end{theorem}

For any partially-bent function $f$ in even dimension, $\dim V(f)$ must be even and $\dim V(f)=0$ if and only if $f$ is bent. So we deduce, from Corollary \ref{sum-mf-quadratic-APN}, that a quadratic function or cubic partially-bent APN function  $F:\mathbb{F}^n\rightarrow \mathbb{F}^n$ must have $2(2^n-1)/3$ bent components if the linear spaces for all components have dimensions $0$ or $2$. Moreover, if there is a component with dimension $2\ell$, $\ell>1$, then the number of bent components has to be increased by $(2^{2\ell}-1)/3-1$ in order for equality of Relation \eqref{mf-quadratic-eqn-1} to hold. Since, in the case of $n=4$, the dimension of linear space of any quadratic function is either $0$ or $2$, then we deduce the following.

\begin{proposition}
	A pure quadratic function $Q:\mathbb{F}^4\rightarrow \mathbb{F}^4$ is APN if and only if there are $10$ bent components.
\end{proposition}

By Equation \eqref{momentum-fourier-equation} in Lemma \ref{momentum-lemma-mf-vectorial}, for any vBf $F:\mathbb{F}^n\rightarrow\mathbb{F}^n$, we have \begin{align}\label{fourier-second-derivatives-APN-eqn} L_4(F)=2^n\sum_{\lambda\neq 0,c,b\in\mathbb{F}^n}\mathcal{F}(D_bD_cF_\lambda).\end{align}
So, by Theorem \ref{APN-momentum-vectorial} and Equation \eqref{fourier-second-derivatives-APN-eqn}, we deduce the following result which relates an APN function to its second order derivatives (this result can also be directly deduced  from Theorem~\ref{APN-first-order-derivatives}).

\begin{theorem}\label{APN-second-order-derivatives}
	Let $F:\mathbb{F}^n\rightarrow \mathbb{F}^n$ be a vBf. Then \[\sum_{\lambda\neq 0,b,c\in\mathbb{F}^n}\mathcal{F}(D_bD_cF_\lambda)\geq 2^{2n+1}(2^n-1).\] Moreover, $F$ is APN if and only if equality holds. 
\end{theorem}

\section{Conclusion}
In this paper, we proved some results about the weight, balancedness and nonlinearity of some splitting functions and a special class of cubic functions. We also proved some results on how the weight and nonlinearity of any Boolean function can be, respectively, related to the weights and nonlinearity of some other functions at a lower dimension. Furthermore, we introduced a parameter of a Boolean function, based on second-order derivatives, from which we derived a characterization of quadratic and cubic APN functions.

\section*{Acknowledgements} The results in this paper appear partially in the last author's MSc thesis and mostly in the first author's PhD thesis, both supervised by the second author.

\end{document}